\documentclass[journal]{IEEEtran}
\usepackage{amssymb}
\usepackage{amsmath}
\usepackage{epsfig}

\newcommand{\be}{\begin{equation}}
\newcommand{\ee}{\end{equation}}
\newcommand{\bea}{\begin{eqnarray*}}
\newcommand{\eea}{\end{eqnarray*}}
\newcommand{\x}{\mathbf{x}}
\newcommand{\y}{\mathbf{y}}
\newcommand{\rr}{\mathbf{r}}
\newcommand{\z}{\mathbf{z}}

\newtheorem{thm}{Theorem}[section]
\newtheorem{lem}[thm]{Lemma}

\newtheorem{rem}[thm]{Remark}

\begin{document}
\title{A remark on the Restricted Isometry Property in Orthogonal Matching Pursuit}
\author{Qun Mo and~Yi Shen
\thanks
{ Research supported in part by the NSF of China under grant
10971189, 11101359, the Doctoral Program Foundation of Ministry of
Education of China under grant 20070335176, Science Foundation of
Chinese University under grant 2010QNA3018 and the Science
Foundation of Zhejiang Sci-Tech University under grant 1113834-Y. }
\thanks
{
Q. Mo is with the Department of Mathematics, Zhejiang University,
Hangzhou, 310027, China (e-mail: moqun@zju.edu.cn ).
}
\thanks
{ Y. Shen is the corresponding author. He is with Department of
Mathematics and Science, Zhejiang Sci-Tech University, Hangzhou
310018, China (e-mail: yshen@zstu.edu.cn). } }

\maketitle

\begin{abstract}
This paper demonstrates that if the restricted isometry constant
$\delta_{K+1}$ of the measurement matrix $A$ satisfies
$$
\delta_{K+1} < \frac{1}{\sqrt{K}+1},
$$
then a greedy algorithm called  Orthogonal Matching Pursuit (OMP)
can recover every $K$--sparse signal $\mathbf{x}$ in $K$ iterations
from $A\x$. By contrast, a matrix is also constructed with the
restricted isometry constant
$$
\delta_{K+1} = \frac{1}{\sqrt{K}}
$$
such that  OMP can not recover some $K$--sparse signal $\mathbf{x}$
in $K$ iterations. This result positively verifies the conjecture
given by Dai and Milenkovic in 2009.
\end{abstract}


\IEEEpeerreviewmaketitle

\section{Introduction}
\IEEEPARstart{C}{ompressive}  sensing is a new type of sampling
theory. It shows that it is highly possible to reconstruct sparse
signals and images from what was previously believed to be
incomplete information \cite{CRT}. Let $\mathbf{x}\in \mathbb{R}^n$
be a signal, we want to recover it from a linear measurement
 \be\label{Axy}
 A \x=\y,
 \ee
where $A$ is a given $m\times n$ measurement matrix. In general, if
$m<n$, the solution of (\ref{Axy}) is not unique. To recover $\x$
uniquely, some additional assumptions on $\x$ and $A$ are needed. We
are interested in the case when $\x$ is sparse. Let $\|\x\|_0$ denote the number of nonzero entries of $\x$.
We say that a vector $\x$ is $K$--\emph{sparse} when $\|\x\|_0\leq K$. To
recover such a signal $\x$, a natural choice is to seek a solution
of the $l_0$ minimization problem
 $$
 \min_{\x} \|\x\|_0 \quad \text{subject to}  \quad A\x=\y
 $$
where $A$ and $\y$ are known.  To  ensure the $K$--sparse solution
is unique, we would like to use the restricted isometry
property  introduced by Cand\`es and Tao in \cite{CT}. A
matrix $A$ satisfies the \emph{restricted isometry
property} of order $K$ with the \emph{restricted
isometry constant} $\delta_K$ if $\delta_K$ is the smallest constant
such that
 $$
 (1-\delta_K)\|\x\|_2^2
 \leq
 \|A \x \|_2^2
 \leq
 (1+\delta_K)\|\x\|_2^2
 $$
holds for all $K$--sparse signal $\x$. If $\delta_{2K}<1$,  the
$l_0$ minimization problem has a unique $K$--sparse solution
\cite{CT}. The $l_0$  minimization problem is equal to the $l_1$
minimization problem when $\delta_{2K}<\sqrt{2}-1$ \cite{C}.
Recently, Mo and Li have improved the sufficient condition to
$\delta_{2K}<0.4931$ \cite{LM}.

OMP is an effective greedy algorithm for seeking the solution of the
$l_0$ minimization problem. Basic references for this method  are
\cite{DMA,PRK} and \cite{T}. For a given $m\times n$ matrix $A$, we
denote  the matrix  with indices of its columns in $\Omega$ by
$A_{\Omega}$. We shall use the same way to deal with the restriction
$\x_{\Omega}$ of the vector $\mathbf{x}$. Let $e_i$ be the $i$th
coordinate unit vector in $\mathbb R^n$. The iterative algorithm
below shows the framework of OMP.

\begin{center}
\begin{itemize}
\item [] \textbf{Input}: $A$, $\y$
\item [] \textbf{Set}: $\Omega_0=\emptyset$, $\rr_0=\y$, $j=1$
\item [] \textbf{while not converge}
\begin{itemize}
\item $\Omega_{j}= \Omega_{j-1} \cup \arg\max_{i} |\langle A e_i,
\rr_{j-1}\rangle| $
\item
$\x_j = \arg\min_\z \| A_{\Omega_j} \z - \y \|_2 $
\item
$\rr_j = \y - A_{\Omega_j} \x_j $
\item
$j = j+1$
\end{itemize}
\item [] \textbf{end while}
\item [] $\hat{\x}_{\Omega_j} = \x_j$,
$\hat{\x}_{\Omega_j^C} = 0$
\item [] \textbf{Return} $\hat{\x}$
\end{itemize}
\end{center}

Davenport and Wakin have proved that
$\delta_{K+1}<\tfrac{1}{3\sqrt{K}}$ is sufficient for OMP to recover
any $K$--sparse signal in $K$ iterations \cite{DW}. Later, Liu and
Temlyakov have improved the condition to
$\delta_{K+1}<\tfrac{1}{(1+\sqrt{2})\sqrt{K}}$ \cite{LT}. By
contrast, Dai and Milenkovic have conjectured that there exist a
matrix  with $\delta_{K+1} \leq \tfrac{1}{\sqrt{K}}$ and a
$K$--sparse vector for which OMP  fails in $K$ iterations. This
conjecture has been confirmed via numerical experiments in \cite{DW}
for the case $K = 2$.
The main results of this paper are consisted by two parts.
\begin{itemize}
\item We prove that
$$
\delta_{K+1}\leq \frac{1}{\sqrt{K}+1}
$$
is sufficient for OMP to exactly recover every $K$--sparse $\x$ in
$K$ iterations.

\item For any given $K\geq2$, we  construct a matrix with
$$\delta_{K+1}=\frac{1}{\sqrt{K}}$$ where OMP fails for at least
one $K$--sparse signal in $K$ iterations.
\end{itemize}

\section{Preliminaries}\label{section1}
Before going further, we introduce some notations.
Suppose $\x$ is a $K$--sparse signal in
$\mathbb{R}^n$. In the rest of this paper, we assume that
$$\x=(x_1,\ldots,x_k,0,\ldots,0)$$
where $x_i\neq 0$, $i=1,2,\ldots,k$, $k\leq K$.  For a given $K$--sparse
signal $\x$ and a given matrix $A$,  we define
$$
S_i := \langle  Ae_i,A\x  \rangle, \quad i=1,\ldots, n.
$$
Denote $ S_0: = \max_{i\in\{1,\ldots, K\}} |S_i| $. The following
lemma is useful in our analysis.

\begin{lem}\label{lem2}
Suppose that the restricted isometry constant $\delta_{K+1}$ of a
matrix $A$ satisfies
 $$
 \delta_{K+1} < \frac{1}{\sqrt{K}+1},
 $$
then $S_0> |S_i|$ for $i>K$.
\end{lem}

\begin{proof}
By  Lemma 2.1 in \cite{C}, we have
 \be\label{eq1}
 |S_i| = |\langle  Ae_i,A\x  \rangle| \leq \delta_{K+1}
 \|\x\|_2\quad
 \text{for all}\ i>K.
 \ee
For the given $K$--sparse $\x$, we obtain
 \bea
 \langle  A\x,A\x  \rangle
 & = & \left \langle  A\sum_{i=1}^K{x_ie_i},A\x \right \rangle \\
 & = & \sum_{i=1}^K  x_i  \langle Ae_i, A\x\rangle \\
 & = & \sum_{i=1}^K  x_i S_i.
 \eea
It follows
 \bea
 (1-\delta_{K+1}) \|\x\|_2^2 & \leq &\langle  A\x,A\x  \rangle \\
 & = & \sum_{i=1}^K  x_i S_i \\
 &\leq & S_0 \|\x\|_1 \\
 &\leq & S_0  \sqrt{K} \|\x\|_2.
 \eea
This implies
 \be\label{eq2}
 \frac{(1-\delta_{K+1})\|\x\|_2}{\sqrt{K}}\leq S_0.
 \ee
It follows from (\ref{eq1}) and (\ref{eq2}) that the lemma holds.
\end{proof}

\section{Main Results}
This section establishes the main results of this paper.

\begin{thm}\label{thm1}
Suppose that $A$ satisfies the restricted isometry
property of order $K+1$ with the
restricted isometry constant
 $$
 \delta_{K+1} < \frac{1}{\sqrt{K}+1},
 $$
then for any $K$--sparse signal $\x$,  OMP will recover $\x$  from
$\y=A\x$ in $K$ iterations.
\end{thm}
\begin{proof}
Consider the first iteration, the sufficient condition for OMP
choosing an index from $\{1,\ldots, K\}$  is
$$
S_0> |S_i |\quad  \text{for all} \  i>K.
$$
By Lemma \ref{lem2}, $\delta_{K+1} < \frac{1}{\sqrt{K}+1}$
guarantees the success of the first iteration. OMP makes an
orthogonal projection in each iteration. By induction, it can be
proved that OMP selects a different index from  $\{1,2,\ldots, K\}$
in each iteration.
\end{proof}

\begin{thm}\label{thm2}
For any given positive integer $K\geq 2$, there exist a $K$--sparse
signal $\x$ and a matrix $A$ with the restricted isometry constant
$$
\delta_{K+1} = \frac{1}{\sqrt{K}}
$$
for which OMP  fails in $K$ iterations.
\end{thm}

\begin{proof}
For any given positive integer $K\geq 2$, let
$$
 A = \left(\begin{array}{cccc}
 &  &  & \frac{1}{K}  \\
 & I_{K}&  & \vdots  \\
 &  &  & \frac{1}{K}  \\
 0 & \ldots& 0 & \sqrt{\tfrac{K-1}{K}}
\end{array}
\right)_{(K+1)\times (K+1)}.
$$
By simple calculation,  we get
$$
 A^T A = \left(\begin{array}{cccc}
 &  &  & \frac{1}{K}  \\
 & I_{K}&  & \vdots  \\
 &  &  & \frac{1}{K}  \\
\frac{1}{K} & \ldots& \frac{1}{K} & 1
\end{array}
\right)_{(K+1)\times (K+1)},
$$
where $A^T$ denotes the transpose of $A$. It is obvious that the
eigenvalues $\{\lambda_i\}_{i=1}^{K+1}$ of $A^TA$ are
$$
\lambda_{1}=\cdots = \lambda_{K-1} =1,\ \lambda_{K}={1-
\tfrac{1}{\sqrt{K}}} \ \text{and}\
\lambda_{K+1}={1+\tfrac{1}{\sqrt{K}}}.
$$
Therefore, the restricted isometry constant $\delta_{K+1}$ of $A$ is
$ \tfrac{1}{\sqrt{K}} $.  Let
$$
\x = (1,1,\ldots,1,0)^T\in \mathbb R^{K+1}.
$$
We have
$$
S_i = \langle  Ae_i, A\x  \rangle =1\quad \text{for all}\ i \in
\{1,\ldots,K+1\}.
$$
This implies OMP fails in the first iteration. Since OMP chooses one
index in each iteration, we conclude that OMP fails in $K$
iterations for the given matrix $A$ and the vector $\x$.
\end{proof}

\begin{rem}
It is challenging to design a measurement matrix having a very small
restricted isometry constant $\delta_{K+1}$; and Theorem \ref{thm2}
shows that this kind of requirement is necessary. However, if we
select multiple indices per iteration, we can recover the
$K$--sparse signal given in Theorem \ref{thm2} in $K$ iterations.
Actually, this technique has been widely used in many related greedy
pursuit algorithms, such as CoSaMP \cite{NT} and Subspace Pursuit
algorithm \cite{DM}.
\end{rem}

\begin{biography}
{Qun Mo} was born in 1971 in China. He has obtained a B.Sc. degree
in 1994 from Tsinghua University, a M.Sc. degree in 1997 from
Chinese Academy of Sciences and a Ph.D. degree in 2003 from
University of Alberta in Canada.

He is current an associate professor in mathematics in Zhejiang University.
His research interests include compressed sensing, wavelet frames and subdivision schemes.
\end{biography}

\begin{biography}
{Yi Shen} was born in 1982 in China. He has obtained a B.S. degree
and a Ph.D. degree in mathematics from the Zhejiang University in
2004 and 2009, respectively.

He is current an associate researcher  in  mathematics in Zhejiang
Sci-Tech University. His research interests include compressed
sensing, wavelet analysis and its applications.
\end{biography}

\end{document}